\documentclass{article}

\usepackage{arxiv}

\usepackage[utf8]{inputenc} % allow utf-8 input
\usepackage[T1]{fontenc}    % use 8-bit T1 fonts
\usepackage{url}            % simple URL typesetting
\usepackage{booktabs}       % professional-quality tables
\usepackage{amsfonts}       % blackboard math symbols
\usepackage{nicefrac}       % compact symbols for 1/2, etc.
\usepackage{microtype}      % microtypography
\usepackage{lipsum}
\usepackage{graphicx}
\usepackage{tikz}
\usetikzlibrary{arrows.meta,  % define arrows head styles
                positioning}  % for nodes positioning
\definecolor{lightblue}{RGB}{173, 216, 230}
\usepackage{algorithm}
\usepackage{algpseudocode}
\usepackage{mathtools}
\usepackage{amsthm}

\usepackage[maxbibnames=99, backend=biber]{biblatex} %Imports biblatex package
\usepackage[colorlinks, linkcolor=blue, urlcolor=]{hyperref}

\title{Improved EFX Approximation Guarantees under Ordinal-based Assumptions\thanks{Accepted at the 22nd International Conference on Autonomous Agents and Multiagent Systems}}

\author{
 Evangelos Markakis \\
 Athens University of Economics and Business \\ and Input Output Global (IOG)\\
 Athens, Greece \\
\texttt{markakis@gmail.com}
\And
Christodoulos Santorinaios \\
 Athens University of Economics and Business \\
 and Archimedes/Athena RC\\
 Athens, Greece\\
\texttt{santgchr@gmail.com}
}

\newcommand{\BibTeX}{\rm B\kern-.05em{\sc i\kern-.025em b}\kern-.08em\TeX}

\newcommand{\aaa}{\mathcal{A}}

\newcommand{\nn}{\mathcal{N}}
\newcommand{\mm}{\mathcal{M}}

\newtheorem{theorem}{Theorem}
\newtheorem{corollary}{Corollary}
\newtheorem{lemma}{Lemma}
\newtheorem{definition}{Definition}

\begin{document}
\maketitle
\begin{abstract}
Our work studies the fair allocation of indivisible items to a set of agents, and falls within the scope of establishing improved approximation guarantees. It is well known by now that the classic solution concepts in fair division, such as envy-freeness and proportionality, fail to exist in the presence of indivisible items. Unfortunately, the lack of existence remains unresolved even for some relaxations of envy-freeness, and most notably for the notion of EFX, which has attracted significant attention in the relevant literature. This in turn has motivated the quest for approximation algorithms, resulting in the currently best known $(\phi-1)$-approximation guarantee by \cite{amanatidis2020multiple}, where $\phi$ equals the golden ratio. So far, it has been notoriously hard to obtain any further advancements beyond this factor. Our main contribution is that we achieve better approximations, for certain special cases, where the agents agree on their perception of some items in terms of their worth. In particular, we first provide an algorithm with a $2/3$-approximation, when the agents agree on what are the top $n$ items (but not necessarily on their exact ranking), with $n$ being the number of agents. To do so, we also study a general framework that can be of independent interest for obtaining further guarantees.
Secondly, we establish the existence of exact EFX allocations in a different scenario, where the agents view the items as split into tiers w.r.t. their value, and they agree on which items belong to each tier. Overall, our results provide evidence that improved guarantees can still be possible by exploiting ordinal information of the valuations.
\end{abstract}

\section{Introduction}

Our work follows the ongoing line of research on fair division with indivisibilities. During the last decade, we have experienced a surge of interest in defining new fairness notions, tailored for allocating a set of indivisible items to a set of agents. This effort has been largely motivated by the realization that the more traditional solution concepts, such as envy-freeness or proportionality are too demanding and fail to exist for the discrete setting.

As a result, there are by now several relaxations that try to incorporate different fairness aspects.
In this work we are particularly interested in the relaxation referred to as EFX (envy-freeness up to any item), but we will also discuss and utilize results for a weaker variant referred to as EF1 (envy-freeness up to one item). The notion of an EFX allocation was defined in \cite{caragiannis2019unreasonable} and demands that an agent $i$ stops being envious of another agent $j$, if at most one item was to be removed from the bundle of agent $j$.   

Over the last years, the question of finding EFX allocations has become one of the most important open problems in fair division \cite{procaccia2020technical}. Despite the effort and interest of the community, we are still not aware if EFX allocations exist for instances with at least $4$ agents. Naturally, this also led to the study of approximation algorithms, which brought forward further insights. But unfortunately, even on this front, it is still unresolved to identify what is the best approximation guarantee that we can have. At the moment, the best known algorithm is by \cite{amanatidis2020multiple}, producing in worst case a $0.618$-EFX allocation, and ideally, it would be very desirable to improve this to a factor as close to 1 as possible.

Driven by these considerations and the difficulty of the problem in its general case, our main focus is to identify conditions under which we can obtain improved approximation guarantees or even exact EFX allocations. We are particularly motivated by a positive result established in \cite{plaut2018almost}, that EFX allocations exist and can be computed efficiently, when all the agents agree on how the items are ranked w.r.t. their value (the agents can have different additive valuation functions, only agreement on the ranking is required). Such instances can be justified as capturing scenarios where all the items are perceived in a similar way by all agents. Although it may seem too strict that the ranking should be exactly the same for all agents and over all the items, this is one of the few positive results that are known on the existence of EFX allocations, and to our surprise, no further progress has been reported if the condition is even minimally violated. 

\subsection{Contribution} We investigate the direction of obtaining further positive results for larger families of instances, by considering relaxations of the {\it common ranking} assumption of \cite{plaut2018almost}. To this end, we investigate two possible such weakenings. Our first main result (Section \ref{subsec:common-top}) is that for instances with $m$ items and $n$ agents, we provide an algorithm with a $2/3$-approximation guarantee, when the set of the top $n$ items is the same for every agent (they do not have to agree on their ranking for them, and we also impose no assumption on the remaining $m-n$ items). This improves the currently known 0.618-approximation for these instances and is particularly useful for cases with a small number of agents, where it can often be the case that they agree on the set of top-tier items. To prove our result, we study a general framework, introduced in \cite{farhadi2021almost}, that is of independent interest for obtaining approximation guarantees. We demonstrate (in Section \ref{subsec:apps}) that indeed, via this framework, we can reprove some already existing results in the literature and also attain further improvements, when agents have additional agreements on the ranking of items. This reveals a smooth degrading of the approximation factor as agents get to have fewer agreements on how they rank the items. Moving on, in Section \ref{sec:tiers}, we obtain a better guarantee by looking at a different relaxation of the common ranking assumption. Namely, we establish the existence and efficient computation of exact EFX allocations in the scenario where the agents view the items as grouped into tiers of size at most three w.r.t. their value, and they agree on the items that belong to each tier. Each tier can be thought of as a different quality/desirability level and the agents in this setting have a common perception on classifying the items to these levels. We note that our proof holds for the more general family of cancelable valuation functions, an interesting superclass of additive functions. Finally, this result also contributes to ruling out instances for constructing potential counterexamples to the non-existence of EFX allocations (which is an ongoing pursuit for the EFX concept).

\subsection{Related work}
The first direct relaxation of envy freeness tailored for indivisible items, was the notion of envy freeness up to 1 good (EF1). Formally it was introduced by Budish in \cite{budish2011combinatorial}, while implicitly considered also in \cite{lipton2004approximately}. Later on, Caragiannis et al. defined the stronger notion of envy freeness up to any good (EFX), in \cite{caragiannis2019unreasonable}. The first positive result in the pursuit of exact EFX allocations, was due to Plaut and Roughgarden, \cite{plaut2018almost}, where the existence of such allocations was proved in the case that all agents have identical valuation functions. 
Making progress even for a small number of agents turned out to be highly non trivial with the currently best result being the existence of EFX allocations for 3 agents with additive valuations, by Chaudhury et al., \cite{chaudhury2020efx}. 

Plaut and Roughgarden also investigated two other ways to attack the problem. The first one was a definition of a multiplicative approximation version of EFX, for which they provided a non-polynomial algorithm with a guarantee of $1/2$. The same ratio was later attained efficiently in \cite{10.5555/3367032.3367053}, and currently, the state of the art approximation is $\phi-1\approx 0.618$, in \cite{amanatidis2020multiple}. The second contribution of Plaut and Roughgarden was a polynomial time algorithm for additive valuations under the restricted setting where agents rank all the items in the same way. After this, exact EFX allocations have also been shown to exist in a variety of other special cases: for additive binary valuations in \cite{10.1007/978-3-030-58285-2_1}, for dichotomous preferences and submodular valuations in \cite{babaioff2021fair}, for two-valued instances and bounded value instances in \cite{amanatidis2021maximum}, and when all agents have one of two possible valuation functions in \cite{mahara2021extension}.

An alternative approach to the problem was initiated in \cite{caragiannis2019envy} where EFX allocations (with high Nash welfare) were shown to exist when some items may be unallocated; giving birth to EFX \textit{with charity}. Subsequent work by Chaudhury et al., \cite{chaudhury2021little}, proved the existence of EFX allocations with bounded charity, i.e., the number of donated items is no more than $n-1$. The bound was improved in \cite{mahara2021extension} to $n-2$, while approximate EFX allocations with sublinear charity was the focus of \cite{chaudhury2021improving}. In that work a connection between the size of the charity and a problem in extremal combinatorics was shown; the second problem was studied outside the context of Fair Division in \cite{berendsohn_et_al:LIPIcs.MFCS.2022.17}, \cite{alon2021divisible} and \cite{meszaros2021zero}. Finally, Berger et al., \cite{Berger_Cohen_Feldman_Fiat_2022}, proved that for the special case of 4 agents, an EFX allocation with at most one unallocated item exists, even for a broader class than the case of additive valuation functions. 

Apart from the notions of EF1 and EFX, there have been already several other fairness concepts that have been proposed, such as the notions of Maximin, Pairwise, and Groupwise Share fairness (MMS, PMMS and GMMS, introduced respectively in \cite{budish2011combinatorial}, \cite{caragiannis2019unreasonable} and \cite{barman2018groupwise}). For an overview of these notions, we refer the reader to \cite{ijcai2022p756}, \cite{aziz2022algorithmic} and the references therein.

%%%%%%%%%%%%%%%%%%%%%%%%%%%%%%%%%%%%%%%%%%%%%%%%%%%%%

\section{Model and Preliminaries}
We consider a set of agents, $N = \{1, \dots, n\}$ and a set $M$ of $m$ indivisible items. Each agent $i
\in N$ is associated with a valuation function $v_i$, which is assumed to be monotone and nonnegative (i.e., we do not allow for chores). 
An allocation $\aaa$ is any valid ordered partition of $M$ into $n$ subsets, $\aaa = (A_1, \dots, A_n)$, where $A_i$ is the bundle of agent $i$. The goal is to compute an allocation that satisfies some desirable fairness criteria.

\subsection{Basic definitions}

Our first set of results in Section \ref{sec:top-items} concerns {\it additive} valuation functions, which is a typical assumption made in the fair division literature. A valuation is additive if every agent $i$ associates a value $v_{ij}$ with each item $j \in M$, and the total value of $i$ for a subset $S\subseteq M$ is given by $v_i(S) = \sum_{j\in S} v_{ij}$. Later in Section \ref{sec:tiers}, we will also define and study generalizations to richer valuations. For ease of notation, we will use $g$ for a singleton set $\{g\}$, so that, e.g., $v_i(g) = v_i(\{g\})$.

Given an input profile, described by the valuation functions of the agents, an ideal solution is to allocate the items so that no one envies someone else's bundle.  
\begin{definition}[Envy freeness-EF]
An allocation $\aaa$ is envy free (EF) if for every pair of agents $i,j$, it holds that $v_i(A_i) \ge v_i(A_j)$.
\end{definition}

Envy freeness turns out to be too strict for indivisible items, and therefore several relaxations have been considered as alternative solutions. 
The first such relaxation is the notion of EF1, due to Budish, \cite{budish2011combinatorial}.

\begin{definition}[EF1]\label{def:ef1}
An allocation $\aaa$ is envy free up to one good (EF1) if for every pair of agents $i,j$, there exists a good $g\in A_j$, such that $v_i(A_i) \ge v_i(A_j \setminus g)$.
\end{definition}

The intuition behind EF1 allocations is that the agents cannot be too envious, in the sense that there always exists a single item whose removal can eliminate envy from one agent to another. This is an efficiently computable fairness notion, as demonstrated in the next subsection. Towards coming closer to envy-freeness, 
Caragiannis et al., \cite{caragiannis2019unreasonable}, defined a stronger notion (but still a weakening of EF), which is also the notion of interest for our work. The difference w.r.t. EF1 is the switch of the quantifiers, so that envy can be eliminated by the removal of any single item.

\begin{definition}[EFX]
An allocation $\aaa$ is envy free up to any good (EFX) if for every pair of agents $i,j$, and for every $g\in A_j$, it holds that $v_i(A_i) \ge v_i(A_j \setminus g)$.
\end{definition}

EFX allocations have turned out to be much harder to compute. As a result, approximate versions of EFX have also received considerable attention. Although there are multiple ways of defining an approximation notion, we will stick to the multiplicative version, as studied in previous works as well: 

\begin{definition}[$\alpha$-EFX]
An allocation $\aaa$ is $\alpha$-EFX, for $\alpha\in [0,1]$, if for every pair of agents $i,j$, and for every $g\in A_j$, it holds that $v_i(A_i) \ge \alpha \cdot v_i(A_j \setminus g)$.
\end{definition}

Hence, our goal is to obtain $
\alpha$-EFX allocations, with $\alpha$ as close to 1 as possible. Furthermore, in the same manner one can also define approximate versions for other concepts (e.g. $\alpha$-EF).
For further illustrations and examples on these concepts, we also refer the reader to \cite{ijcai2022p756}.

\subsection{Envy cycle elimination}

In the sequel we will use as a building block, a well-known algorithm, which we will refer to as Envy Cycle Elimination (ECE). Introduced in \cite{lipton2004approximately}, it computes an EF1 allocation in polynomial time, even for general valuations. 
The algorithm uses a graph-theoretic approach, where each agent is viewed as a node, and where envy from an agent $i$ towards an agent $j$ corresponds to the directed edge $(i, j)$. Hence, any allocation (not necessarily of the whole set of goods), can be represented by the corresponding directed graph, referred to as the envy graph of the allocation and denoted by $E_G$. The algorithm starts from the empty allocation and allocates one item per round. In each round, if there exists an unenvied agent $i$ in the current allocation, then she can receive the next item without violating the EF1 property. If no such agent exists, then any agent in $E_G$ has an incoming edge, meaning that the graph contains at least one directed cycle in the form $i_1 \to i_2 \to \dots \to i_1$. Reallocating the bundles backwards along the cycle, i.e., agent $j$ would receive $A_{j+1}$, removes some envy edges, and by repeating the process enough times, we can eliminate all cycles and leave some agent without an incoming edge. This would be the agent to receive the item of the current round (breaking ties arbitrarily if there are multiple such agents). The formal description of the algorithm follows.

\begin{algorithm}[H]
\caption{Envy Cycle Elimination$(\nn, \mm)$}\label{alg:ECE}
\begin{algorithmic}[1]
\State{Set $A_i = \emptyset$ for every agent $i$}
\While{$\exists$ some unallocated item $g$}
\State{Pick some source agent $s$ (guaranteed to exist, break ties arbitrarily) }
\State{Set $A_s = A_s \cup \{g\}$}
\State{Decycle the envy graph by repeatedly finding envy cycles and reallocating backwards the bundles along each cycle}\label{alg1:line5}
%\State Allocate $g$ to one of the source agents 
\EndWhile
\State\Return{$\aaa = (A_1, \dots, A_n)$} 
\end{algorithmic}
\end{algorithm}

\begin{theorem}[implied by \cite{lipton2004approximately}]
Algorithm \ref{alg:ECE} computes an EF1 allocation in polynomial time.
\end{theorem}
At this point, we should note that in the original version in \cite{lipton2004approximately}, it suffices to execute the decycling step in line 5, until we obtain at least one source agent (with no incoming edges), to maintain the EF1 property. For our purposes, it is convenient to eliminate all cycles, and maintain the invariant that $E_G$ is a DAG (Directed Acyclic Graph) at the end of each round.  We will discuss more about the source agents of $E_G$ in Section \ref{sec:tiers}.

%%%%%%%%%%%%%%%%%%%%%%%%%%%%%%%%%%%%%%%%%%%%%%%%%%%%%%%%%%%%%%%%%%%%%%%%

\section{Agreement on top items}
\label{sec:top-items}

In this section, we will focus on additive valuations and on extending the positive result of Plaut and Roughgarden, \cite{plaut2018almost}, 
regarding the existence of EFX allocations, when the agents share a common ranking over the items. A natural path to follow is to study instances where this assumption is somewhat relaxed (e.g., truncated to apply only to a few items). Although our results do not 
guarantee existence of exact EFX allocations, we do obtain improved approximation guarantees compared to the state of the art for some of 
these instances, and our general aim is to study the tradeoff between the approximation ratio and the extent of departure from the common ranking assumption over all items. 

\subsection{General approximation framework}
Before addressing our main goal, we will first introduce a framework, originally presented in \cite{farhadi2021almost}, for producing approximation algorithms. This framework serves both as a unifying umbrella for some of the already existing results, but also helps us in establishing our improvements. We stress that from now on, we will assume that $m> n$, since otherwise, there exists a trivial exact 
EFX allocation, where each agent receives at most one item.

We start with a few observations to develop some intuition. Consider a variation of the ECE algorithm, where during the first $n$ rounds, every time we select a source agent to allocate the next item to, we actually let her pick her favorite one, among the unallocated items. Assume also, we select a distinct agent in each of these first $n$ rounds. The remainder of the ECE algorithm can run as presented in the previous section. Interestingly, Algorithm 2 of \cite{10.5555/3367032.3367053} achieves a $1/2$-EFX approximation by essentially doing this, even though its first phase may look unrelated to ECE\footnote{The algorithm of \cite{10.5555/3367032.3367053} was developed with a different fairness notion in mind and uses perfect matchings, with the first one corresponding to the right of choice in the first $n$ rounds of ECE.}. It is instructive to present the proof below under this viewpoint, for the sake of completeness, and for motivating the more general framework that we want to exploit in the sequel.

\begin{theorem}[Implied by \cite{10.5555/3367032.3367053}]
\label{thm:1/2}
The Envy Cycle Elimination algorithm computes a $1/2$-EFX allocation when in each of the first $n$ rounds, the source agent (selected in line 3) is granted the right to choose her favorite unallocated item.
\end{theorem}

\begin{proof}
Consider two agents $i$ and $j$ and let $A_i, A_j$, be the final bundles allocated to these agents at the end of the algorithm. We will prove that $i$ satisfies the 1/2-EFX condition w.r.t. $j$. Let $g_i$ be the first item that was allocated to agent $i$. Since the ECE algorithm never decreases the valuation of any agent during its rounds, we have that $v_i(A_i) \ge v_i(g_i)$. Also, let $h$ be the last item that was added to the bundle $A_j$ of agent $j$. We know that when $h$ was added, the owner of the bundle $A_j\setminus h$ at that time, was unenvied, thus $v_i(A_i) \ge v_i(A_j\setminus h)$. 
Moving on, if $A_j$ contains only $h$, then $i$ trivially satisfies the EFX condition w.r.t. $j$. Hence, suppose $|A_j|\geq 2$, which means that $h$ was picked after the $n$-th round.
Then, from the definition of $g_i$, it follows $v_i(h) \le v_i(g_i) \le v_i(A_i)$. Adding the two inequalities for $v_i(A_i)$, yields $2v_i(A_i) \ge v_i(A_j)$.
\end{proof}

We note that the previous result actually produces a $1/2$-EF allocation when $m> n$, hence, an even better guarantee than $1/2$-EFX. 
A crucial observation arising from this proof is that the more times an agent $i$ gets to pick an item before becoming envious of another agent, the better the approximation. In particular, we used in the proof of Theorem \ref{thm:1/2} the inequality $v_i(A_i) \ge v_i(h)$, but if $i$ had picked an item $k$ times during the execution of ECE, and before the addition of $h$ to the bundle $A_j$, we could replace this with: 
$v_i(A_i) \ge k\cdot v_i(h)$. And this would lead to a 
$\frac{k}{k+1}$-EFX approximation. 
Although it may not always be easy to enforce such a property, it still yields an approach for obtaining approximation guarantees for certain families of instances, as we shall see. In fact, all we need is to be able to produce first a partial EFX allocation of some items, that satisfies such a property for every agent. We can now formally put everything together and have the following scheme. 

\begin{algorithm}
\caption{General approximation framework}\label{alg:gaf}
\begin{algorithmic}[1]
\State For $\alpha, \beta >0$, compute a partial $\alpha$-EFX allocation $\mathcal{S} = (S_1,\dots S_n)$, with the property that $$ v_i(S_i) \ge \beta \cdot v_i(h) 
\text{ for all $i\in N$ and all } h \in M\setminus \bigcup_{j\in N} S_j$$
\State Continue with running the ECE algorithm, until there are no unallocated items 
\end{algorithmic}
\end{algorithm}

\begin{theorem}[\cite{farhadi2021almost}]
\label{thm:framework}
Algorithm \ref{alg:gaf} computes a $\min\left(\alpha, \frac{\beta}{\beta +1}\right)$-EFX allocation. Moreover, if the partial allocation $\mathcal{S}$ is $\gamma$-EF1 for some $\gamma\leq 1$, or if it can be computed efficiently, the same properties carry over for the final allocation as well. 
\end{theorem}
\noindent For the sake of completeness (and since the proof of Theorem 3 is not provided in \cite{farhadi2021almost}), we include a proof below.

\begin{proof}
Let $\aaa= (A_1,\dots, A_n)$ denote the final output of the algorithm and fix an agent $i$. As a first case, consider an agent $j$, who receives as her final bundle, 
one of the bundles of the allocation $\mathcal{S}$, say $S_j$. Then, we know that $v_i(A_i) \geq v_i(S_i)$, and since $\mathcal{S}$ was an $\alpha$-EFX allocation, agent $i$ satisfies the $\alpha$-EFX condition towards agent $j$. 
For the second case, consider an agent $j$, so that $A_j$ is a 
strict superset of some bundle of $\mathcal{S}$. Let $h$ be the last item added to $A_j$ by the algorithm. This means that the owner of the bundle $A_j\setminus h$, right before $h$ was allocated, was unenvied. 
Hence, if $B_i$ was the bundle owned by $i$ at that time, we have that 
$v_i(B_i) \ge v_i(A_j\setminus h)$.  
Moreover, we also know by the algorithm construction that
$v_i(B_i) \ge v_i(S_i) \ge \beta\cdot v_i(h)$. 
Multiplying the first inequality by $\beta$ and adding it to the second yields 
$$(\beta+1)v_i(B_i) \geq \beta v_i(A_j) \Rightarrow v_i(B_i) \ge \frac{\beta}{\beta +1} v_i(A_j)$$ 
The same inequality holds for the bundle $A_i$ too, since the ECE algorithm never makes an agent worse. Thus, $i$ satisfies the 
$\frac{\beta}{\beta+1}$-EF condition, which implies the same for EFX as well. Combining the two cases completes the proof.  
\end{proof}

For an immediate illustration, the proof of Theorem \ref{thm:1/2} corresponds to applying the framework with $\alpha=1$ and $\beta=1$. We note also that all the results presented in the next subsections, satisfy EF1 exactly (Theorem \ref{thm:framework} applies with $\gamma=1$), and we will therefore focus only on the achieved EFX approximation.

\subsection{An algorithm under a common top-$n$ set}
\label{subsec:common-top}

Suppose that the agents have additive valuations, and they agree on which are the $n$ most favorable items, without necessarily agreeing on their ranking w.r.t. their value\footnote{Existence of possible ties with the remaining $m-n$ items does not affect our result.}. 
A simple thought to achieve a $2/3$ approximation is by trying to apply the framework of the previous subsection, with $\alpha = 2/3$ and $\beta=2$. Before presenting our algorithm, we provide first some initial ideas and intuition. To achieve $\beta=2$ for a partial allocation, it is natural to consider allocating first 2 items per agent. This could be done by allocating the first $n$ items in a round robin fashion, and then using the reverse order of round robin for agents to select a second item. 
Each agent will indeed prefer each of her two items to any unallocated one, thus we can guarantee that $\beta = 2$ in the conditions of Theorem \ref{thm:framework}. Unfortunately though, this partial allocation may not be $2/3$-EFX and hence we cannot use our framework for analyzing it. A violation of this property means that for some pair of agents $i,j$, it holds $v_i(S_i) < 2/3\cdot v_i(S_j \setminus g)$. Our algorithm is built on exploiting that information, and on constructing a more careful partial allocation; some agents will receive only one item (they will be content with one item from the top-$n$ set) whereas the remaining ones will receive two items (one of which will again be from the top-$n$ set). 

\begin{algorithm}
\caption{2/3 EFX for identical top-$n$ set}\label{alg:efx23set}
\begin{algorithmic}[1]
\State{Let $M = T \cup B$, where $T$ is the set of the $n$ most valuable items and $B$ consists of the remaining $m-n$ items}
%\State{Construct $n$ identical chips $c$}
\For{each agent $i$}
\State{$h_i = \arg\max\limits_{m \in T} v_i(m)$}\label{alg3:line3}
\State{$g_{1,i} = \arg\min\limits_{m \in T} v_i(m)$}\label{alg3:line4}
\State{$g_{2,i} = \arg\max\limits_{m \in B} v_i(m)$}\label{alg3:line5}
\If{$v_i(g_{1,i}) + v_i(g_{2,i}) \ge \frac{2}{3}\cdot v_i(h_i)$}\label{alg3:line6}
\State{$A_i = \{g_{2,i}\}$ (non-content agent)}\label{alg3:line7}
\State{$B = B \setminus g_{2,i}$}\label{alg3:line8}
\Else
\State{$A_i = \{h_i\}$ (content agent)}\label{alg3:line10}
\State{$T = T \setminus h_i$}\label{alg3:line11}
\EndIf
\EndFor
\State{For every non-content agent, allocate to her one item arbitrarily from $T$}\label{alg3:line14}
\State{Continue with running the ECE algorithm} 
\end{algorithmic}
\end{algorithm}

\begin{theorem}\label{th:23set}
Algorithm \ref{alg:efx23set} efficiently computes a $2/3$-EFX allocation, when the agents agree upon the set of top $n$ items.
\end{theorem}

\begin{proof}
We will discern two cases based on whether an agent was content or not, as determined by the condition in line \ref{alg3:line6}. Before proceeding, note that all the items in $T$ will eventually get allocated. Every content agent gets one item from $T$ in line \ref{alg3:line10}. And then, there will be as many left-over items of $T$ as the number of non-content agents, and each of them receives one such item in line \ref{alg3:line14}. 

To prove our result using the framework of Theorem \ref{thm:framework}, we will establish the desired properties for $\alpha=2/3$ and $\beta = 2$. To be more precise, let $S = (S_1,\dots, S_n)$ be the partial allocation produced up until the execution of line \ref{alg3:line14} of the algorithm, and let $M'$ denote the set of unallocated items at that time. By Theorem \ref{thm:framework}, it suffices to establish that right before the execution of the ECE algorithm, $\mathcal{S}$ is a $2/3$-EFX allocation, and that also $v_i(S_i) \geq 2 v_i(m)$ for any $m\in M'$ and every $i\in N$. 

Fix an agent $i$ and for simplicity in the analysis, we will drop the subscript $i$ and use $h$ instead of $h_i$ and so on. We have the following two cases. 

\noindent {\bf Case 1:} $v_i(g_1) + v_i(g_2) < \frac{2}{3}\cdot v_i(h)$ (agent $i$ is content).
    
\noindent We have that agent $i$ receives exactly one item in the allocation $\mathcal{S}$, $S_i = \{h\}$, thus every other agent satisfies the EFX condition towards her. At the same time, $i$ also satisfies the EFX condition towards agents with one item. Hence consider an agent $j$ having obtained two items. We know that $i$ has received her favorite remaining item from $T$ in line \ref{alg3:line10}, therefore she cannot envy the item of $T$ that was received by $j$ in line \ref{alg3:line14}. Agent $i$ also cannot envy the other item of $j$ since it belongs to $B$ and is not more valuable than the items of $T$. In conclusion, agent $i$ satisfies the EFX condition towards all other agents under allocation $\mathcal{S}$. Moreover, it holds that
    \begin{align*}
        v_i(g_1) + v_i(g_2) &\ge 2\cdot v_i(g_2) \implies
        v_i(g_2) \le \frac{1}{3}\cdot v_i(h)\\
        \intertext{And by the definition of $g_2$}
        v_i(h) &\ge 3\cdot v_i(m)\ \forall m \in M'
    \end{align*}
    
    Thus all conditions of Theorem \ref{thm:framework} that concern agent $i$ are satisfied.

\noindent {\bf Case 2:} $v_i(g_1) + v_i(g_2) \ge \frac{2}{3}\cdot v_i(h)$ (non-content agent $i$).
    
\noindent Now agent $i$ will receive initially the item $g_2$ from $B$, and later, in line \ref{alg3:line14}, she will receive an item from $T$, which we will denote by $h'$. Since $g_1$ was the least valued top item for $i$, it follows that $v_i(h') \ge v_i(g_1)$. Therefore, the bundle of $i$, $S_i = \{h', g_2\}$, is at least as valuable as $\{g_1, g_2\}$, and thus $v_i(S_i) \geq \frac{2}{3}\cdot v_i(h)$. Similar to Case 1, it suffices to check the envy of $i$ towards any agent $j$ receiving two items. As in Case 1, the item $h$ has to be more valuable than any of the two items received by $j$. Thus, for any $g\in S_j$,
     
     $$ v_i(S_j \setminus g) \le v_i(h) \implies v_i(S_i) \ge \frac{2}{3} \cdot v_i(S_j \setminus g)$$
     meaning that $i$ satisfies $2/3$-EFX towards $j$. Finally, 
     $$v_i(S_i) \ge 2\cdot v_i(g_2) \ge 2\cdot v_i(m)\ \forall m \in M'$$
     Thus the conditions of Theorem \ref{thm:framework} are again satisfied.

Putting everything together, the partial allocation after line \ref{alg3:line14} is at worst $2/3$-EFX. Also, for every agent $i$ and every unallocated good $m$, it holds that $v_i(A_i) \ge 2\cdot v_i(m)$, and the proof is completed.
\end{proof}
\subsection{Other applications of the framework}
\label{subsec:apps}
We close this section by presenting some more straightforward applications of our framework based on other classes of restricted scenarios, generalizing some results from the existing literature. 

\begin{corollary}[Relaxed top ranking]\label{corrank}
Assume that for some $\ell \leq m$, all agents agree not only on the set of the top $\ell$ items, but also on their order w.r.t. their value. Then one can efficiently compute a $\frac{k}{k+1}$-EFX allocation, with $k = \lfloor \ell/n \rfloor$.
\end{corollary}

\begin{proof}
Given $\ell$, we know that a partial EFX allocation for the top $\ell$ items can be computed efficiently by \cite{plaut2018almost}, due to the identical ranking assumption for them. Hence, we can achieve $\alpha=1$ in our framework. It suffices to establish Theorem \ref{thm:framework}, with $\beta=k= \lfloor \ell/n \rfloor$. 
Let $\mathcal{S} = (S_1,\dots, S_n)$ be the allocation of the top $\ell $ items. Note that for each agent with a bundle of size no less than $k$, the relation we want follows directly due to additivity. 
Assuming that some agent $i$ has less than $k$ items, then by the pigeonhole principle, some other agent $j$ has more than $k$ items. Since the allocation $\mathcal{S}$ is EFX, $v_i(S_i) \ge v_i(S_j\setminus g)$ for any $g\in S_j$. But now, the bundle in the right hand side is of size at least $k$ and again the relation we want regarding the unallocated items follows.  
\end{proof}

The above result captures the trade-off between the achievable approximation ratio and the degree of agreement on the ranking of the top items: for every additional group of $n$ items that the agents agree upon, the approximation ratio improves. 
Also, it demonstrates that different applications of the framework for the same setting can yield suboptimal results; for $k=1 (\ell = n)$, the above corollary yields a $1/2$ ratio, whereas we know that we can get a better approximation by the previous subsection, using different values for $\alpha$ and $\beta$.

The next corollary combines the approximation framework with an algorithm presented in \cite{amanatidis2021maximum}, that computes an exact EFX allocation when all agents value all items within a bounded interval of the form $[x, 2x]$. We can obtain the following generalization of \cite{amanatidis2021maximum}.

\begin{corollary}[Relaxed bounded interval]
Assuming that all agents value their top $\ell$ items within an interval of the form $[x, 2x]$, then one can efficiently compute a $\frac{k}{k+1}$-EFX allocation with $k = \lfloor \ell/n \rfloor$.
\end{corollary}\label{corbound}
\begin{proof}
The proof is identical with that of Corollary \ref{corrank} except that $\mathcal{S}$ is constructed via the Modified Round Robin algorithm of \cite{amanatidis2021maximum}.
\end{proof}

Finally, we also present a result in the opposite direction, where we assume full disagreement among the agents on top items. Our framework can help us in improving the state of the art and obtain a $2/3$-approximation in the following case.
 
\begin{corollary}[Distinct top items]
\label{cor:distinct}
Assuming that each agent $i$ has a different favorite good $f_i$, then a $2/3$-EFX allocation can be computed efficiently.
\end{corollary}\label{cordis}

\begin{proof}
We construct a partial EFX allocation $\mathcal{S}$ as follows: we first let each agent get her favorite item, $f_i$, and then run $n$ rounds of the ECE algorithm, to allocate $n$ more items, but as in Theorem \ref{thm:1/2}, we grant each agent the right to choose her favorite unallocated item, say $g_i$, when selected as the unenvied agent. Therefore, in the produced partial allocation $\mathcal{S}$, we have that $S_i = \{f_i, g_i\}$. Then
$$ v_i(S_i) \ge v_i(f_i) \ge v_i(S_j\setminus g) ~~\forall g\in S_j$$ where the last inequality is due to the fact that $S_j \setminus g$ is a singleton set and $f_i$ is the best item according to $i$. Hence, $\mathcal{S}$ is a partial EFX allocation ($\alpha=1$). Proving that Theorem \ref{thm:framework} is satisfied with $\beta=2$ is identical to Case 2 in the proof of Theorem \ref{th:23set}.
\end{proof}

%%%%%%%%%%%%%%%%%%%%%%%%%%%%%%%%%%%%%%%%%%%%%%%%%%%%%%%%%%%%%%%%%%%%%%%%

\section{Tiered rankings and further improvements}
\label{sec:tiers}
Our second main result concerns the existence of exact EFX allocations. Following the spirit of the previous results, one can ask if we can prove even better results, when the agreement between agents goes on beyond the first tier of most valuable items. Keeping in mind that we are also interested in relaxations of the common ranking assumption over all items, studied in \cite{plaut2018almost}, a natural approach is to consider scenarios where the goods can be split into tiers w.r.t. their value. This leads to the following definition.
\begin{definition}\label{deftier}
An instance of our problem has a common \textit{tiered ranking} among all agents, if there exists an ordered partition of all items $M = (M_1, M_2, \dots, M_\ell)$, such that for every agent $i$,
$$ \forall g\in M_k, \forall h \in M_{j > k}:\ v_i(g) \ge v_i(h)$$
Moreover, we define the size of the tiered ranking to be the size of the largest tier, i.e., $\max_{j\in [\ell]} |M_j|$.
\end{definition}

Under the prism of tiers, the identical ranking setting can be viewed as a tiered ranking of size 1. We will study such instances, under a more general family of valuation functions, defined below.

\begin{definition}[Definition 2.1 in \cite{Berger_Cohen_Feldman_Fiat_2022}]\label{defcan}
A valuation function $v$ is cancelable if for any bundles $S,T \subset M$, and item $g \in M \setminus (S\cup T)$, it holds that $$ v(S\cup g) > v(T \cup g) \implies v(S) > v(T)$$
\end{definition}

Clearly, cancelable valuations include additive ones and are monotone (as we have no chores). Moreover, they include other well known classes, such as multiplicative and unit-demand valuations (see \cite{Berger_Cohen_Feldman_Fiat_2022}). So far, there are several results for additive valuations that have been shown to carry over to cancelable functions as well, making this class a natural generalization of additivity (this is not applicable for our results in Section \ref{sec:top-items} however). In the following theorem, we assume at least 3 agents, since the existence for 2 agents is already well known. 

\begin{theorem}
\label{thm:tiers}
Assuming $n\ge 3$, and that the agents have cancelable valuations with a common tiered ranking of size at most 3,
then an EFX allocation exists and can be computed efficiently.
\end{theorem}

We stress that although the assumption of tiers with up to three items may look like a simple relexation of the common ranking assumption over all items, the proof turns out to be much more involved, than the existence result of \cite{plaut2018almost}. To proceed, we state the following simple lemma on cancelable valuations, that we will use.
\begin{lemma}\label{lem:ineq}
Let $S,T,Q \text{ and }R$ be sets such that $S \cap Q =\emptyset$, and $T \cap R =\emptyset$. Then, for a cancelable valuation function $v$,
$$\begin{rcases}v(S) \ge v(T)\\ v(Q) \ge v(R) \end{rcases} \implies v(S \cup Q) \ge v(T \cup R)$$
\end{lemma}

\begin{proof}
In the definition of cancelable valuations, it is easy to see that one direction implies the opposite: $v(T) \le v(S) \implies v(T \cup g) \le v(S\cup g)$. Applying this for every $g \in Q \setminus T$ gives
 $$v(T\cup Q) = v(T\cup (Q \setminus T)) \le v(S\cup (Q \setminus T)) \le v(S \cup Q)$$ where the last inequality is due to monotonicity. Similarly we obtain $v(Q \cup T) \ge v(R \cup T)$, and the lemma follows.

\end{proof}

\begin{proof}[Proof of Theorem 4.3]
Under our assumption, we can partition the items into tiers so that all agents have preferences in the following form: 
%$$ a_1, b_1, c_1 \succeq a_2, b_2, c_2 \succeq ... \succeq a_k, b_k, c_k$$
$$ M_1 \succeq M_2 \succeq ... \succeq M_s$$
where each $M_i$ is a set of items with $|M_i|\leq 3$, and the notation $M_i \succeq M_j$ means that all the agents find the items of $M_i$ at least as valuable as those of $M_j$. We note that within a tier, it is not necessary that all agents rank the items in the same way. Also, some of the tiers may only contain just one or two items. 

To prove the theorem, we need to establish that for any ordered pair of agents $(i, j)$, agent $i$ satisfies the EFX condition towards $j$. For convenience in the analysis, whenever this is violated, we will say that agent $i$ {\it strongly envies} agent $j$. 

We will prove the theorem by induction on the number of tiers. We present the proof here only for the case where each tier has exactly 3 items, since the other cases are easier. Note that the basis of our induction, allocating the first triplet, is easy: just give a single item to three different agents, and this is trivially EFX. For the inductive step, assume that we have a partial EFX allocation after allocating all the goods up to some tier $k$, and we will need to see how to allocate the three items of tier $k+1$. We will denote from now on the three items of tier $k+1$ by $a, b, c$. 

Let $E_G$ be the envy graph of the allocation that has resulted from the induction hypothesis for the items of the first $k$ tiers. We can assume that we can maintain $E_G$ from tier to tier as a DAG. If it is not, we can always remove all the envy cycles prior to continuing with the next tier, by running the decycling step (line \ref{alg1:line5}) of the ECE algorithm. Hence, $E_G$ has at least one source. We discern three cases based on the number of sources in $E_G$, and discuss them in order of difficulty.
\smallskip

\noindent \textbf{Case 1:} $E_G$ has at least three source agents.

\noindent This is the easiest case since we have three items and we can pick three sources (i.e., unenvied agents) and just do a matching. The EFX property is maintained.

On a high level, both remaining cases work as follows: the agents closest (in the sense of topological distance) to the sources are possible new sources. Therefore, if our current sources outvalue them or maybe get their bundles via some envy cycle elimination, we can allocate the remaining items. Otherwise, some current source will receive more than one items without violating the EFX property. 
We introduce also some extra notation. We will often refer to agents based on their level when we view $E_G$ with the following topological ordering: sources are at level 0, the agents envied by some source and with no envy within them at level 1, and so on (at each level, incoming edges come only from previous levels). Let also $(A_1,\dots, A_n)$ be the current EFX allocation. 

\smallskip 
\noindent \textbf{Case 2:} $E_G$ has one source agent $s_1$.
    
\noindent Let $o_1, o_2, \dots, o_\ell$ be the agents of level 1 ordered based on $s_1$'s valuation: $v_{s_1}(A_{o_1}) \le v_{s_1}(A_{o_2}) \le \dots \le v_{s_1}(A_{o_\ell})$.
    
\noindent $\bullet$ \underline{Subcase 2a:} $s_1$ can receive multiple items without breaking EFX.
        
\noindent If $s_1$ can receive all three goods of the tier without violating the EFX property we are done. If not, but she still can receive two goods, say $a$ and $b$, we allocate them to her. It remains to allocate $c$. If $v_{s_1}(A_{s_1} \cup a \cup b) \ge v_{s_1}(A_{o_1})$, we allocate $c$ to $o_1$ and complete the inductive step. Indeed, note that $c$ is the least valuable item of the bundle that $o_1$ has now. Since $s_1$ was not envying $o_1$ after she received $a$ and $b$, the EFX property will not be violated by giving $c$ to $o_1$. There is no other pair of agents that we need to check since no one else is allocated any items.  
Suppose now that $v_{s_1}(A_{s_1} \cup a \cup b) < v_{s_1}(A_{o_1})$. We also know that since $s_1$ could not receive all 3 items, some agent, say $x$, must envy her. The fact that $s_1$ was the single source of $E_G$ means that $x$ is reachable from $s_1$ via an envy path. 
Therefore, the allocation of $a$ and $b$ to $s_1$ has created an envy cycle $s_1 \to o_i \to \dots \to x \to s_1$, for some agent $o_i$ of level 1. After decycling the graph (using line \ref{alg1:line5} of the ECE algorithm), $s_1$ will be in possession of the bundle $A_{o_i}$. Since previously she was the only one envying the bundle, we can now allocate $c$ to her without disrupting EFX.
        
\noindent $\bullet$ \underline{Subcase 2b:} $s_1$ cannot receive multiple items.

\noindent If allocating 2 items is problematic for the EFX property, we deduce that after giving, say $a$, to $s_1$, some agent $x$ becomes envious of her. Firstly, we will identify possible new sources. Those are the $o_i$ agents, and more specifically $o_1$ and $o_2$, and agents envied only by $o_1$ (and maybe $s_1$), which we denote by $t_1, t_2,\dots$, and so on. The nodes of interest are shown in Figure \ref{fig:subcase2b}.
        
    \begin{figure}[h] 
	\centering
    \begin{tikzpicture}[
node distance = 12mm and 9mm,
     C/.style = {circle, draw, thick, fill=lightblue,
                 minimum size = 8mm,
                 label=#1},
every edge/.style = {draw, -Straight Barb},
                        ]
\node[font=\large] (n1) [C=below:]  {$s_1$};
\node[font=\large] (n2) [C=below:, right=of n1]  {$o_1$};
\node[font=\large] (n3) [C=below:, above=of n2]  {$o_2$};
\node[font=\large] (n4) [C=below:, right=of n2]  {$t_1$};
\path   (n1) edge   (n2)
        (n2) edge   (n4);
\path   (n1) edge   (n3);
\path   (n1) edge[dashed, bend left] (n4);
    \end{tikzpicture}
    \caption{Subcase 2b}
    \label{fig:subcase2b}
    \end{figure}
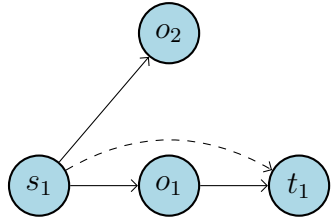
    Similarly to the previous case, if $v_{s_1}(A_{s_1}\cup a) \ge v_{s_1}(A_{o_2})$, we allocate one item to both $o_1$ and $o_2$ and we are done. Otherwise, we continue in the same spirit as before. Suppose first that $x$ is reachable from $o_r$, for some $r\ge 2$. Then, given that $v_{s_1}(A_{s_1}\cup a) < v_{s_1}(A_{o_2}) \leq v_{s_1}(A_{o_r})$, an envy cycle $s_1 \to o_r \to \dots \to x \to s_1$ is created. After decycling it, we allocate one item to $o_1$ and one to the current owner of $o_2$ (it could be either $o_2$ or $s_1$), and we are done. Suppose now that $r=1$. We have that $x$ could be some agent only reachable from $o_1$ or she could be $o_1$ herself, and we need to look at these two further cases separately.
    
    \noindent When $r=1$, and $x$ can be some agent other than $o_1$, the role of the $t_i$ nodes becomes more clear. So far, the possible new sources were always $o_1$ and $o_2$. Now they are $o_1$ and some node $t_i$. To proceed, note that there may not be a path from $s_1$ to $x$ after allocating item $a$ to $s_1$, but we will reallocate the bundles as if there was; checking that the EFX property is maintained is easy. Also, we pick $x$ as the agent furthest away from $s_1$. The image looks as follows:
        
        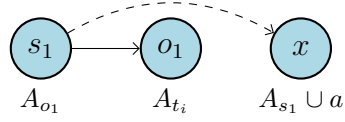
\begin{figure}[h] 
	\centering
    \begin{tikzpicture}[
node distance = 12mm and 9mm,
     C/.style = {circle, draw, thick, fill=lightblue,
                 minimum size = 8mm,
                 label=#1},
every edge/.style = {draw, -Straight Barb},
                        ]
\node[font=\large] (n1) [C=below:$A_{o_1}$]  {$s_1$};
\node[font=\large] (n2) [C=below:$A_{t_i}$, right=of n1]  {$o_1$};
\node[font=\large] (n3) [C=below:$A_{s_1} \cup a$, right=of n2]  {$x$};

\path   (n1) edge   (n2);
\path   (n1) edge[dashed, bend left] (n3);
    \end{tikzpicture}
    \caption{Subcase 2b with $r=1$ and $x\neq o_1$}
    \label{fig:subcase2bi}
    \end{figure}
    
Now $s_1$ with her new bundle is again a source (single one if $x$ is not a source) and if any node other than $x$ is envious of $A_{o_1} \cup b$ or $A_{o_1} \cup c$, the next source will be either $o_1$ (owning $A_{t_i}$) or $o_i, i\ge 2$ and we are done. The same applies if $s_1$ can stop envying some node $t_i$ after she receives a good due to $t_i$ becoming the final source or if $s_i$ will receive $A_{t_i}$ after some envy cycle elimination. If neither is true and $s_1$ cannot receive both of the remaining items without violating the EFX property, $x$ must be a source and some agent reachable only from her envies $s_1$ after she receives a good. In this scenario, if one of the two matchings between the two sources and the two items produce an EFX allocation, we have completed this case. Otherwise, there is a cycle containing both sources and nodes from one or both connected components (Lemma \ref{lem:ineq} guarantees that no agent between $s_1$ and $x$ can strongly envy the other one) if we substitute $a$ with $b$ or $c$. Now, $s_1$ is in possession of some $A_{t_i}$ and she can have the last item.

When $r=1$, and $x$ is $o_1$, we have only one source and one possible new source. However, apart from $o_1$, no other agent would strongly envy $A_{s_1} \cup a \cup b$ (or any other combination of $A_{s_1}$ with a pair of goods from the given tier); otherwise we would be back to $x$ being different that $o_1$. Therefore, we will compensate the lack of possible sources by possibly adding two items to one bundle. We start by asking $s_1$ to choose between receiving her favorite item or $o_1$'s favorite, thus creating an envy cycle of size 2 and causing a swap. Note however, that if $s_1$ and $o_1$ have a different favorite item, then $s_1$ will always choose to swap since she may well get her favorite item right after, and Lemma \ref{lem:ineq} guarantees the optimality of the choice. In any case, and with $a$ the favorite item of $o_1$, the owner of $A_{o_1}$ will be the new source. If she can get both of the remaining items, the proof is completed. Otherwise some agent envies $A_{o_1} \cup b$ or $A_{o_1} \cup c$. Since we have a single source, that envious agent is reachable so we ask our source to choose between her favorite item or the one that forms the envy cycle implied above. In either case, the new source will be eligible to receive the last item even if her bundle is $A_{s_1} \cup a$. 

\smallskip

\noindent\textbf{Case 3:} $E_G$ has two source agents, $s_1$ and $s_2$.
     
\noindent Now, we have more sources than Case 2, but also a harder time identifying the possible third one to receive an item. To bypass this problem we partition the envy graph $E_G$ in the following manner:
$$ E_G = s_1 \cup s_2 \cup V_1 \cup V_2 \cup V_{12}$$
where $V_1$ (resp. $V_2$) is the set of nodes reachable \textit{only} from $s_1$ (resp. $s_2$) via an envy path and $V_{12}$ is the set of nodes reachable from both sources. Since $E_G$ is a DAG, the same follows for the graphs induced by $V_1, V_2$ and $V_{12}$ as well. If $s'_1$ is a source of the $V_1$ DAG we have that $s_1$ is the only agent envious of her; otherwise it would be reachable from $s_2$ thus contradicting the definition of the partition. Therefore, $s'_1$ is a possible new source substituting $s_1$ and, symmetrically, a source $s'_2$ of $V_2$ is a candidate substitution for $s_2$. A source $s'_{12}$ of $V_{12}$ may be a possible replacement for both.
Now, if $s_1$ can receive two items or one and simultaneously stop envying $s'_1$, we easily allocate all the items of the tier. Therefore, we assume this is not the case and denote by $e_1$ the agent envious of $s_1$, after receiving one item of the current tier, and respectively, $e_2$ for $s_2$. We will do some case analysis based on which connected component, $e_1$ and $e_2$ belong to. Since they may not be unique, we define $T_1$ to be the set of agents envying $s_1$ (and respectively $T_2$ for $s_2$). Fortunately, due to symmetry, the number of different subcases is small.

\noindent$\bullet$\underline{Subcase 3a:} $T_1 \cap V_1 \neq \emptyset$ (or resp. $T_2 \cap V_2 \neq \emptyset$).
         
\noindent Let $e_1 \in T_1 \cap V_1$. In that case $e_1$ is reachable from $s_1$ and we have an envy cycle. After applying a decycling step, $s_1$ will own $A_{s'_1}$, a bundle of which she previously was the only envious agent. Therefore, we can allocate one more item to $s_1$ and the other to $s_2$.
         
\noindent$\bullet$\underline{Subcase 3b:} $T_1 \cap (s_2 \cup V_2) \neq \emptyset$ (or resp. $T_2 \cap (s_1 \cup V_1) \neq \emptyset$).
         
\noindent Let $e_1 \in s_2 \cap V_2$. Now, $e_1$ is not reachable from $s_1$. However, the same must apply to $s_2$ and $e_2$, otherwise we are back to subcase 3a . In other words, $e_2 \in E_G \setminus V_2$. This means that, after allocating two of the three items of the tier, a cycle in the form $s_1 \to \dots\to e_2 \to s_2 \to\dots \to e_1\to s_1$, is created. Once we decycle the graph, checking that the EFX property is maintained is trivial. Moreover, there will be again two sources. If one of them does not own one of the two bundles previously owned by the sources, we can allocate the final item. In the unique case where $e_1 = s_2$ and $s_1 = e_2$ were the only choices for envious agents, we can allocate the last item to any of them. To see why, assume that $s_1$ owns $A_{s_2} \cup a \cup c$. Agent $s_2$ cannot strongly envy her and if some other agent does so, it means that $e_1 = s_2$ was not the only choice.
     
After careful inspection, subcases 3a and 3b cover for 8 out of the 9 possible scenarios. It remains to check the case where both envious agents are in $V_{12}$.
     
\noindent$\bullet$\underline{Subcase 3c:} $T_1 \cap V_{12} \neq \emptyset$ and $T_2 \cap V_{12} \neq \emptyset$.
         
\noindent Let $e_1, e_2 \in V_{12}$. Note that since $e_1$ is reachable from $s_1$, if the first node in the path between them belongs to $V_1$, the argument of subcase 3a applies. Therefore, the first envy edge in the path is from $s_1$ to some source of $V_{12}$. If $e_1$ and $e_2$ belong to different weakly connected components of $V_{12}$, we are done since one of the owners ($s_1$ or $s_2$) of some $A_{s'_{12}}$ will get the last item. Assuming the contrary, the image is given in Figure \ref{fig:subcase3c}. 
         
         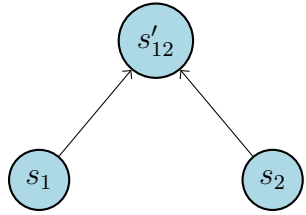
\begin{figure}[h] 
	\centering
    \begin{tikzpicture}[
node distance = 12mm and 9mm,
     C/.style = {circle, draw, thick, fill=lightblue,
                 minimum size = 8mm,
                 label=#1},
every edge/.style = {draw, -Straight Barb},
                        ]
\node[font=\large] (n1) [C=below:]  {$s_1$};
\node[font=\large] (n2) [C=below:, above right=of n1]  {$s'_{12}$};
\node[font=\large] (n3) [C=below:, below right=of n2]  {$s_2$};
\path   (n1) edge   (n2);
\path   (n3) edge   (n2);
    \end{tikzpicture}
    \caption{Subcase 3c}
    \label{fig:subcase3c}
    \end{figure}
         
    We now select the source, say $s_1$, and the item, say $a$ to allocate, based on $e_1$'s preferences, who we pick to be in maximum topological distance (this can be achieved by checking all possible cases of allocating one item from the tier to $s_1$). After reallocating the bundles along the cycle (if $s_1$ stops envying $s'_{12}$, she becomes a possible new source similar to case 2 and we are done accordingly), the situation now is depicted in Figure \ref{fig:subcase3c-2}.
    
    \begin{figure}[h] 
	\centering
    \begin{tikzpicture}[
node distance = 12mm and 9mm,
     C/.style = {circle, draw, thick, fill=lightblue,
                 minimum size = 8mm,
                 label=#1},
every edge/.style = {draw, -Straight Barb},
                        ]
\node[font=\large] (n1) [C=below:{$A_{s_2}$}]  {$s_2$};
\node[font=\large] (n2) [C=left:{$A_{s'_{12}}$}, above=of n1]  {$s_1$};
\node[font=\large] (n3) [C=below:{$A_{s_1} \cup a$}, right=of n1]  {$e_1$};
\path   (n1) edge   (n2);
\path   (n1) edge[dashed]   (n3);
    \end{tikzpicture}
    \caption{Subcase 3c continued}
    \label{fig:subcase3c-2}
    \end{figure}
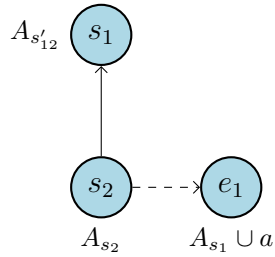
The way we picked $e_1$, we have that $e_1$ and every possible node reachable from her cannot strongly envy $A_{s_2} \cup b \cup c$; otherwise the strongly envious agent would have been envious before and in a greater topological distance\footnote{To be precise, $e_1$ could have been envious of $A_{s_2} \cup b$ but then we restart the tier allocation working with $s_2$.}. Therefore, either $s_2$ will get both items or there will be an envious agent she can reach. In the end, whoever owns $A_{s'_{12}}$ will get the last item and the proof is completed.
         
\end{proof}

We close this section, by providing one more positive result, but under an assumption in the opposite direction of Theorem \ref{thm:tiers}, in the same spirit as Corollary \ref{cor:distinct} in Section \ref{sec:top-items}.

\begin{theorem}[Distinct top tiers]
Assuming that each agent has a different favorite tier $T_i$ of size $\lfloor m/n \rfloor$, then an EFX allocation exists. 
\end{theorem}

\begin{proof}
We show how to construct the allocation. Let $m = kn + r, 0\le r < n$,  and let each $T_i$ be of size $\lfloor m/n \rfloor = k$. We allocate each $T_i$ to the corresponding agent, and to $r$ of them one of the remaining items arbitrarily. If agent $i$'s envy towards agent $j$ violates the EFX property, it means that $i$ envies a set of cardinality at most $k$, which contradicts the definition of $T_i$. Thus, the allocation is EFX. 
\end{proof}

Note that the above result holds even for general (monotone) valuations.

\section{Conclusions}
We have presented a unifying framework for obtaining approximate EFX allocations. This framework allowed us to reprove and extend some existing results of the literature and at the same time obtain improved approximations for families of instances where agents agree only on their perception of what are the most valuable items. For instances with a relatively large number of goods, this imposes quite minimal assumptions on the structure of the instance. Finally, we have also established the existence of exact EFX allocations by considering assumptions on how the agents rank the whole set of items w.r.t their value, extending the existence result of \cite{plaut2018almost}, that holds under the common ranking assumption.

The most intriguing question of whether exact EFX allocations exist for arbitrary additive valuations still remains open. Along with this, identifying the best possible approximation guarantee that we can have in polynomial time is also unresolved. Overall, we view our work as adding a positive note on these directions, providing evidence that improved approximations and existence results on certain families of instances can still be attainable.

%%%%%%%%%%%%%%%%%%%%%%%%%%%%%%%%%%%%%%%%%%%%%%%%%%%%%%%%%%%%%%%%%%%%%%%%

\section*{Acknowledgements}
This research was supported by the Hellenic Foundation for Research and Innovation, by the “1st Call for HFRI Research Projects to support faculty members
and researchers and the procurement of high-cost research equipment” (Project Num. HFRI-FM17-3512). Part of this work was carried out while the second author was a M.Sc. student at the National Technical University of Athens (NTUA) and we acknowledge the use of facilities and resources of NTUA.

\printbibliography %Prints bibliography

\end{document}